\documentclass[11pt, twocolumn,twoside]{IEEEtran}

\usepackage{cite}      

\usepackage{graphicx}  
\usepackage{amssymb}
\usepackage{verbatim}
\usepackage{graphicx}
\usepackage{amsmath}
\usepackage{url}

\usepackage{amsthm}
\makeatletter
\def\th@plain{%
  \thm@notefont{}
  \itshape 
}
\def\th@definition{%
  \thm@notefont{}
  \normalfont 
}
\makeatother

\newtheorem{thm}{{\bf Theorem}}
\newtheorem{coro}[thm]{{\bf Corollary}}
\newtheorem{lem}[thm]{{\bf Lemma}}

\begin{document}
%
\title{New Conditions for Sparse Phase Retrieval}
\author{Mehmet~Ak\c{c}akaya,~\IEEEmembership{Member,~IEEE,}
        ~and~Vahid~Tarokh,~\IEEEmembership{Fellow,~IEEE}
\thanks{M. Ak\c{c}akaya is
with Beth Israel Deaconess Medical Center, Harvard Medical School, Boston, MA, 02215. Vahid Tarokh is with the School of Engineering and Applied Sciences, Harvard University, Cambridge, MA, 02138. (e-mails: makcakay@bidmc.harvard.edu, vahid@seas.harvard.edu).  M. Ak\c{c}akaya would like to acknowledge grant support from NIH K99HL111410-01.}}
\markboth{DRAFT}{DRAFT}
\maketitle

\begin{abstract}
We consider the problem of sparse phase retrieval, where a $k$-sparse signal ${\bf x} \in {\mathbb R}^n \textrm{ (or } {\mathbb C}^n\textrm{)}$ is measured as ${\bf y} = |{\bf Ax}|,$ where ${\bf A} \in {\mathbb R}^{m \times n} \textrm{ (or } {\mathbb C}^{m \times n}\textrm{ respectively)}$ is a measurement matrix and $|\cdot|$ is the element-wise absolute value.
For a real signal and a real measurement matrix ${\bf A}$, we show that $m = 2k$ measurements are necessary and sufficient to recover ${\bf x}$ uniquely.
For complex signal ${\bf x} \in {\mathbb C}^n$ and ${\bf A} \in {\mathbb C}^{m \times n}$, we show that $m = 4k-2$ phaseless measurements are sufficient to recover ${\bf x}$. It is known that the multiplying constant $4$ in $m = 4k-2$ cannot be improved.
\end{abstract}

\IEEEpeerreviewmaketitle

\section{Introduction} \label{sec:intro}

\PARstart{L}{et} ${\bf x} \in {\mathbb H}^n,$ and ${\bf A} \in {\mathbb H}^{m \times n},$ where ${\mathbb H}$ is either the real field ${\mathbb R}$ or the complex field ${\mathbb C}.$  We use $| \cdot|$ to denote element-wise absolute value operation from ${\mathbb H}$ to ${\mathbb R}$.

Phase retrieval deals with the problem of estimating a signal from phaseless measurements,
\begin{equation} \label{meas-model}
 {\bf y} = |{\bf Ax}|.
 \end{equation}

This can be used to model problems in x-ray crystallography \cite{Millane}, diffractive imaging \cite{Bunk}, astronomical imaging \cite{Dainty} and medical imaging \cite{Zoroofi}, where the measurement matrix ${\bf A}$ is typically the Fourier matrix.
There exists some efforts in the literature for reconstruction of such signals \cite{Chen, Liu, GS, Fienup}.

For a general measurement matrix ${\bf A}$, sufficient conditions for unique phase retrieval were established \cite{Balan}. Balan et al have shown that $m \geq 2n - 1$ generic measurements are sufficient for unique retrieval (up to global phase) in ${\mathbb H} = {\mathbb R}$, and $m \geq 4n-2$ measurements suffice in ${\mathbb H} = {\mathbb C}$ \cite{Balan}. Furthermore, it was recently shown that semidefinite programming can be used to reconstruct signals if ${\bf A}$ is a Gaussian random matrix, as long as $m \geq c_0 n \log n$ for a sufficiently large constant $c_0$ \cite{CSV}. Alternative methods have also been explored for practical reconstruction with theoretical guarantees \cite{CESV, Mallat, Sanghavi}.

Recently, there has been an interest in sparse phase retrieval, where the number of non-zero coefficients (or $\ell_0$ norm) of ${\bf x}$, denoted $||{\bf x}||_0$, is much smaller than the dimensionality $n.$ This a-priori knowledge about the signal can be used to reduce the number of measurements in practice \cite{Moravec, Shechtman, Ohlsson}. Theoretical analysis of sparse phase retrieval has also been performed in certain scenarios. For instance, $O(k\log(n/k))$ measurements were shown to be sufficient for stable sparse phase retrieval over ${\mathbb R}$ \cite{Eldar-Mendelson}. This is the same order as in compressed sensing \cite{Candes, Donoho} where linear measurements with phase are available.

The conditions for exact sparse reconstruction, with no measurement noise, is well-understood in the context of compressed sensing/sparse approximation, where $m = 2k$ is necessary and sufficient \cite{Candes-Romberg-Tao-fft}, with polynomial-time reconstruction algorithms for Vandermonde-based measurement matrix designs (which also include partial Fourier transforms that sample only the low-frequency components) achieving this bound \cite{AkTar}. The corresponding bounds for exact sparse phase retrieval is a relatively new area of research. In a paper that describes a computationally feasible algorithm for sparse phase retrieval \cite{Voroninski}, the results of \cite{Balan} were applied in the sparse case for a discussion of injectivity. This straightforward extension, used for illustrative purposes of injectivity, implied $4k-1$ measurements were sufficient for unique sparse phase retrieval of a $k$-sparse signal in ${\mathbb R}$, and $8k-2$ were sufficient for a $k$-sparse signal in ${\mathbb C}$. The real case was also characterized in \cite{Eldar-pr} with the same result.

In this note, we study the exact sparse phase retrieval problem for $k$-sparse ${\bf x} \in {\mathbb H}^n.$ For ${\mathbb H} = {\mathbb R},$ we show that $m = 2k$ measurements are sufficient to recover every $k$-sparse signal from phaseless measurements, for ${\bf A},$ whose rows are a generic choice of vectors in ${\mathbb R}^n$. In conjunction with the results from compressed sensing, this forms a necessary and sufficient condition. For ${\mathbb H} = {\mathbb C}$, we show that $m = 4k-2$ measurements are sufficient to recover every $k$-sparse signal from phaseless measurements if the rows of ${\bf A}$ are a generic choice of vectors in ${\mathbb C}^n$. The outline of the paper is given next. We state our results for the real and complex cases in Section \ref{sec:results}. The proofs are given in Section \ref{sec:proof}, where we use a combinatorial approach to extend our previous coding theory-based work \cite{AkTar} for the real case, and we extend the proof technique of \cite{Balan} in the complex case. We make our conclusions and provide directions for future research in Section \ref{sec:conc}.


\section{Problem Statement and Main Results} \label{sec:results}


In sparse phase retrieval, the reconstructor solves
\begin{equation} \label{opt-ell0}
 \min ||{\bf x}||_0 \:\:\: \textrm{s. t.} \:\:\:  {\bf y} = |{\bf Ax}|.
 \end{equation}

We aim to characterize the number of sufficient measurements, $m$ in terms of the sparsity $k$ (and possibly the dimensionality of the sparse signal $n$) for which there is a unique solution to the optimization problem in (\ref{opt-ell0}), up to global phase. We state our results separately for the real and complex cases.
\begin{thm} For ${\bf A} \in {\mathbb R}^{m \times n},$ whose rows are a generic choice of vectors in ${\mathbb R}^{n},$ $m \geq 2k$  measurements are sufficient to guarantee unique phase retrieval for any $k$-sparse signal ${\bf x} \in {\mathbb R}^n.$ 
 \end{thm}

Here, a generic choice of frame vectors indicate a dense Zariski-open set \cite{Balan}. Interestingly, the number of sufficient measurements for unique sparse phase retrieval from phaseless measurements
in the real case match the number of necessary and sufficient measurements for unique sparse signal recovery from linear measurements
\cite{Candes-Romberg-Tao-fft,AkTar}. This leads to the following corollary:
\begin{coro} To guarantee the unique sparse phase retrieval for every $k$-sparse ${\bf x} \in {\mathbb R}^n$, $m = 2k$ measurements are necessary and sufficient.
\end{coro}

In the complex case, we have the following result:
\begin{thm} For ${\bf A} \in {\mathbb C}^{m \times n},$ whose rows are a generic choice of vectors in ${\mathbb C}^{n},$ $m \geq 4k-2$ measurements are sufficient to guarantee unique phase retrieval for any $k$-sparse signal ${\bf x} \in {\mathbb C}^n.$ 
 \end{thm}

%

We note that the bound in the complex case is analogous (in terms of the number of non-zero elements) to the non-sparse phase retrieval problem \cite{Balan}, whereas in the real case, one additional measurement is needed.


\section{Proof of Main Results} \label{sec:proof}
\subsection{Notation}
We define ${\mathbb T}(\mathbb H) = \{x \in {\mathbb H} : |x| = 1\}.$ The space of admissible diagonal phase matrices is defined as
\begin{align}
{\cal P}({\mathbb H}) = &\{ {\bf P} \in ({\mathbb T}(\mathbb H) \cup \{0\})^{m \times m} : \nonumber \\
		& p_{ij} = 0 \:\: \forall i \neq j \textrm{ and } p_{ii} \neq p_{jj} \textrm{ for some }i \neq j \nonumber \\
		& \textrm{and } |p_{ii}| = 1 \:\: \forall i \} \nonumber
\end{align}
where $p_{ij}$ is the $(i,j)^{\textrm{th}}$ element of ${\bf P}.$ Note the condition $p_{ii} \neq p_{jj} \textrm{ for some }i \neq j,$ ensures any admissible phase matrix is not a multiple of the identity matrix, which would only change the global phase.

Let ${\bf a}_k$ be the $k^{\textrm{th}}$ column of ${\bf A}$, and ${\bf a}^{(k)}$ be the $k^{\textrm{th}}$ row. We define ${\bf A}_{\cal J}$ to be matrix whose columns are $\{{\bf a}_j: j \in {\cal J}\}.$ Similarly we define ${\bf A}^{({\cal J})}$ to be the matrix whose rows are $\{{\bf a}^{(j)}: j \in {\cal J}\}.$ We let $[l] = \{1, 2, \dots, l\}$ for any positive integer $l.$ For any matrix ${\bf B},$ let ${\bf B}^T$, ${\bf B}^*$ and ${\cal N}({\bf B})$ denote the transpose, conjugate transpose and right null space of ${\bf B}$. $|\cdot|$ denotes element-wise absolute value for a vector as in Section \ref{sec:intro}, and also denotes cardinality for an index set.


\subsection{Proof for the Real Case} \label{sub:real-proof}

We extend the coding theory-based proof technique from \cite{AkTar}. To this end, we define the phase-generalized minimum distance of a matrix as follows:

\textbf{Definition.} For any ${\bf C} \in {\mathbb R}^{m \times n}$ with $m<n$, we define the phase-generalized minimum distance as the smallest integer greater than
$$ \min_{\substack{{\cal I}, {\cal J}: \\ |{\cal I}| + |{\cal J}| \leq m}} \min_{{\bf P} \in {\cal P}({\mathbb R})} \textrm{rank}\Big( \big[{\bf C}_{\cal I} , {\bf P}{\bf C}_{\cal J}\big]\Big).$$

By rank considerations, the phase-generalized minimum distance of any $m \times n$ matrix is $\leq m + 1.$ Now we state a necessary and sufficient condition for recovery in this case:
\begin{lem} \label{lem31}
Let ${\bf A} \in {\mathbb R}^{m \times n}$ with $m<n$ have phase-generalized minimum distance $d.$ Then $\ell_0$ minimization can uniquely recover any vector ${\bf x}$ from phaseless measurements ${\bf y} = |{\bf Ax}|$ if $||{\bf x}||_0 \leq \lfloor(d-1)/2\rfloor.$ Furthermore, this condition is necessary. 
\end{lem}

\begin{proof} Suppose there are two distinct solutions ${\bf x}$ and ${\bf z}$ such that $||{\bf x}||_0 = ||{\bf z}||_0$ (since both can be picked by $\ell_0$ minimization) and
$${\bf y} = |{\bf Ax}| = |{\bf Az}|.$$
In other words
\begin{equation} \label{eq1-proof1}
{\bf r} = {\bf Ax} = {\bf PAz},
\end{equation}
for some ${\bf P} \in {\cal P}({\mathbb R}).$ Let ${\cal I}$ and ${\cal J}$ be the supports of ${\bf x}$ and ${\bf z}$ respectively. Then Equation (\ref{eq1-proof1}) can be re-written as:
\begin{equation}
\big [{\bf A}_{\cal I} , {\bf P}{\bf A}_{\cal J}\big]   \left[ \begin{array}{cc}
{\bf x}_{\cal I}\\
{\bf z}_{\cal J}  \end{array} \right]\  = {\bf 0}.
\end{equation}

If $||{\bf x}_{\cal I}||_0 = ||{\bf z}_{\cal J}||_0 \leq \lfloor(d-1)/2\rfloor,$ then the vector $[{\bf x}_{\cal I}^T, \: {\bf z}_{\cal J}^T]^T$ has dimensionality $\leq 2 \lfloor(d-1)/2\rfloor < d-1.$ By the definition of phase-generalized minimum distance, any such $[{\bf A}_{\cal I} , {\bf P}{\bf A}_{\cal J}\big] $ has full rank, implying the vector on the right should be ${\bf 0},$ leading to a contradiction.

However, if $||{\bf x}_{\cal I}||_0 = ||{\bf z}_{\cal J}||_0 \geq \lfloor(d-1)/2\rfloor + 1,$ then $[{\bf x}_{\cal I}^T, \: {\bf z}_{\cal J}^T]^T$ has dimensionality $\geq 2 (\lfloor(d-1)/2\rfloor + 1) \geq d.$ Since the rank of $[{\bf A}_{\cal I} , {\bf P}{\bf A}_{\cal J}\big] $ is at most $d-1,$ this implies there are infinitely many solutions for $[{\bf x}_{\cal I}^T, \: {\bf z}_{\cal J}^T]^T$, violating uniqueness.
\end{proof}

We next use a combinatorial approach to show if the rows of ${\bf A}$ are chosen generically in ${\mathbb R}^n$, then it has phase-generalized minimum distance $2k + 1$ if $m = 2k.$ We look at different cases based on the possible overlap between the columns of ${\bf A}$, corresponding to the supports of distinct sparse vectors as solutions to the optimization problem in (\ref{opt-ell0}):

\subsubsection{No Support Overlap} \label{sec231}
Let ${\cal I}, {\cal J} \subset [n]$ be two index sets of cardinality $k$ each such that ${\cal I} \cap {\cal J} = \emptyset$. Consider the $m \times 2k$ matrix
$${\bf M}({\cal I}, {\cal J}, {\mathbf P}) = \big[{\bf A}_{\cal I} , {\bf P}{\bf A}_{\cal J}\big],$$
for some ${\bf P} \in {\cal P}({\mathbb R}).$ If the rows of ${\bf A}$ are chosen generically in ${\mathbb R}^n$, then $\textrm{rank}\big({\bf M}({\cal I}, {\cal J}, {\mathbf P})\big) = 2k$ for $m \geq 2k.$ Since there are finitely many choices for ${\cal I}$, ${\cal J}$ and ${\mathbb P},$ a union bound argument extends the result to all possible choices of ${\bf P}$, and $k$ columns of ${\bf A}$. Then Lemma \ref{lem31} implies that there do not exist two distinct sparse vectors with no support overlap, whose $\ell_0$ norms are $\leq k$, that map to the same phaseless measurements.

\subsubsection{Full Support Overlap} \label{sec232}
Consider an index set ${\cal I} \subset [n]$ with cardinality $k$. Let ${\bf C} = {\bf A}_{\cal I}.$ Consider the $m \times 2k$ matrix
$${\bf M}({\cal I}, {\mathbf P}) = \big[{\bf C}, {\bf P}{\bf C}\big],$$
for some ${\bf P} \in {\cal P}({\mathbb R}).$ Suppose there exist distinct 
${\bf u}_1, {\bf u}_2 \in {\mathbb R}^k$ such that
\begin{equation} \label{eq2-proof1}
{\bf M}({\cal I}, {\mathbf P})   \left[ \begin{array}{cc}
{\bf u}_1\\
{\bf u}_2 \end{array} \right]\  = {\bf 0}.
\end{equation}

Let ${\cal R}_{1} = \{k: p_{kk} = 1\}$ and ${\cal R}_{-1} = \{k: p_{kk} = -1\}.$ We note neither set is empty by the definition of ${\cal P}({\mathbb R}).$ Let the cardinality of ${\cal R}_{1}$ be $l$, hence that of ${\cal R}_{-1}$ is $m - l$. Then $[{\bf u}_1^T, \: {\bf u}_2^T]^T$ is in ${\cal N}\big(({\bf M}({\cal I}, {\mathbf P}))^{({\cal R}_{1})}\big) \cap {\cal N}\big(({\bf M}({\cal I}, {\mathbf P}))^{({\cal R}_{-1})}\big).$

Suppose the the rows of ${\bf A}$ are chosen generically in ${\mathbb R}^n,$ and $\min(l, m-l) \geq k$. Then with probability 1, if $l \geq m - l$ (hence $l \geq k$), then
$${\cal N}\big(({\bf M}({\cal I}, {\mathbf P}))^{({\cal R}_{1})}\big) = \bigg\{ \left[ \begin{array}{cc}
{\bf u}_{1}\\
{\bf u}_{2}  \end{array} \right]\  \: : \: {\bf u}_1 = -{\bf u}_2 \bigg\},$$
which is not in ${\cal N}\big(({\bf M}({\cal I}, {\mathbf P}))^{({\cal R}_{-1})}\big)$ unless ${\bf u}_1 = {\bf u}_2 = {\bf 0}$.
Otherwise if $m-l \geq l$ (hence $m-l > k$), then
$${\cal N}\big(({\bf M}({\cal I}, {\mathbf P}))^{({\cal R}_{-1})}\big) = \bigg\{ \left[ \begin{array}{cc}
{\bf u}_{1}\\
{\bf u}_{2}  \end{array} \right]\  \: : \: {\bf u}_1 = {\bf u}_2 \bigg\},$$
which is not in ${\cal N}\big(({\bf M}({\cal I}, {\mathbf P}))^{({\cal R}_{1})}\big)$ unless ${\bf u}_1 = {\bf u}_2 = {\bf 0}$.
Thus ${\cal N}\big(({\bf M}({\cal I}, {\mathbf P}))^{({\cal R}_{1})}\big) \cap {\cal N}\big(({\bf M}({\cal I}, {\mathbf P}))^{({\cal R}_{-1})}\big) = \{{\bf 0}\},$ and thus ${\bf u}_1 = {\bf u}_2 = {\bf 0},$ leading to a contradiction. 

%

Noting there are finitely many choices for ${\cal I}$ and ${\mathbb P},$ a union bound argument shows that there do not exist two distinct sparse vectors with full support overlap that both map to the same phaseless measurements as long as $m \geq 2k -1.$

We also note that the result for the full support overlap case can be generated from the results of \cite{Balan} for the non-sparse phase retrieval problem.

\subsubsection{Partial Support Overlap} \label{sec233}
Let ${\cal I}, {\cal J} \subset [n]$ be two index sets of cardinality $k$ such that ${\cal I} \cap {\cal J} \neq \emptyset$. Let $|{\cal I} \cap {\cal J}| = w.$ Then $|{\cal I} \backslash {\cal J}| = |{\cal J} \backslash {\cal I}| = k - w.$ Consider the $m \times 2k$ matrix
$${\bf M}({\cal I}, {\cal J}, {\mathbf P}) = \big[{\bf A}_{{\cal I} \backslash {\cal J}}, {\bf A}_{{\cal I} \cap {\cal J}}, {\bf P}{\bf A}_{{\cal J} \backslash {\cal I}}, {\bf P}{\bf A}_{{\cal I} \cap {\cal J}}\big].$$

As in Section \ref{sec232}, we let ${\cal R}_{1} = \{k: p_{kk} = 1\}$ and ${\cal R}_{-1} = \{k: p_{kk} = -1\},$ with $|{\cal R}_{1}| = l.$ Without loss of generality, we assume the first $l$ rows correspond to ${\cal R}_{1},$ and re-write ${\bf M}({\cal I}, {\cal J}, {\mathbf P})$ as
\begin{equation} \nonumber
{\bf M}({\cal I}, {\cal J}, {\mathbf P}) =
\left[ \begin{array}{cccc}
{\bf A}_{{\cal I} \backslash {\cal J}}^{({\cal R}_{1})} & {\bf A}_{{\cal I} \cap {\cal J}}^{({\cal R}_{1})} & {\bf A}_{{\cal J} \backslash {\cal I}}^{({\cal R}_{1})} & {\bf A}_{{\cal I} \cap {\cal J}}^{({\cal R}_{1})} \\
{\bf A}_{{\cal I} \backslash {\cal J}}^{({\cal R}_{-1})} & {\bf A}_{{\cal I} \cap {\cal J}}^{({\cal R}_{-1})} & -{\bf A}_{{\cal J} \backslash {\cal I}}^{({\cal R}_{-1})} & -{\bf A}_{{\cal I} \cap {\cal J}}^{({\cal R}_{-1})}
\end{array} \right]
\end{equation}
We will characterize the rank of this matrix, or equivalently that of
\begin{equation} \label{eq-overlap1}
{\bf M}' =
\left[ \begin{array}{cccc}
{\bf A}_{{\cal I} \backslash {\cal J}}^{({\cal R}_{1})} & {\bf A}_{{\cal I} \cap {\cal J}}^{({\cal R}_{1})} & {\bf A}_{{\cal J} \backslash {\cal I}}^{({\cal R}_{1})} & {\bf 0} \\
{\bf A}_{{\cal I} \backslash {\cal J}}^{({\cal R}_{-1})} & {\bf 0} & -{\bf A}_{{\cal J} \backslash {\cal I}}^{({\cal R}_{-1})} & {\bf A}_{{\cal I} \cap {\cal J}}^{({\cal R}_{-1})}
\end{array} \right]
\end{equation}

If $l \geq 2k-w$ or $m-l \geq 2k-w,$ then the first $l$ rows or the last $m-l$ rows are trivially full-rank, implying that ${\bf M}'$ and thus ${\bf M}({\cal I}, {\cal J}, {\mathbf P})$ is full-rank. Thus, we only consider $l < 2k-w$ and $m-l < 2k-w.$ This implies $2k-w > m-l \geq 2k-l,$ i.e. $w < l.$ Similarly, $w < m-l.$ We first consider the following sub-matrix from Equation (\ref{eq-overlap1}) (after re-arranging the columns):
$${\bf C} \triangleq [{\bf A}_{{\cal I} \cap {\cal J}}^{({\cal R}_{1})}, {\bf A}_{{\cal I} \backslash {\cal J}}^{({\cal R}_{1})} , {\bf P}{\bf A}_{{\cal J} \backslash {\cal I}}^{({\cal R}_{1})}].$$
This matrix is equivalent to
\begin{equation} \label{eq-overlap2}
\left[ \begin{array}{cc}
{\bf I}_{[w]} & {\bf 0} \\
{\bf 0} & {\bf C}'
\end{array} \right]
\end{equation}
where
$${\bf C}' = {\bf C}_{[2k-w] \backslash [w]}^{([l] \backslash [w])} - {\bf C}_{[w]}^{([l] \backslash [w])} \Big({\bf C}_{[w]}^{([w])}\Big)^{-1} {\bf C}_{[2k-w] \backslash [w]}^{([w])}.$$

\vspace{0.1cm}
A similar procedure can be followed for the sub-matrix (after re-arranging) $${\bf D} \triangleq [{\bf A}_{{\cal I} \cap {\cal J}}^{({\cal R}_{-1})}, {\bf A}_{{\cal I} \backslash {\cal J}}^{({\cal R}_{-1})},  -{\bf A}_{{\cal J} \backslash {\cal I}}^{({\cal R}_{-1})}]$$ to yield an equivalent matrix analogous to that in Equation (\ref{eq-overlap2}), where ${\bf D}'$ is defined analogously. Then ${\bf M}({\cal I}, {\cal J}, {\mathbf P})$ is equivalent to
\begin{equation} \nonumber
\left[ \begin{array}{cc}
{\bf I}_{[2w]} & {\bf 0} \\
{\bf 0} & {\bf B}
\end{array} \right]
\end{equation}
where
${\bf B} =
\left[ \begin{array}{c}
{\bf C}' \\
{\bf D}'
\end{array} \right].$ For a generic choice of rows for ${\bf A},$ the rank of ${\bf B}$ is $2k-2w$ and that of ${\bf M}({\cal I}, {\cal J}, {\mathbf P})$ is $2k$. With the union bound procedure, this is true for all ${\cal I}, {\cal J}$ and ${\bf P}.$ 


\vspace{0.3cm}
Combining all the results from Sections \ref{sec231}, \ref{sec232} and \ref{sec233}, we have that if the rows of ${\bf A} \in {\mathbb R}^{m \times n}$ are chosen generically and $m \geq 2k$ then the phase-generalized minimum distance of ${\bf A}$ is $2k + 1.$ Based on Lemma \ref{lem31}, this implies that the $\ell_0$ minimization can uniquely recover every $k$-sparse ${\bf x}$ from phaseless measurements ${\bf y} = |{\bf Ax}|.$


\subsection{Proof for the Complex Case}

The union bound technique for the support selection used in Section \ref{sub:real-proof} does not extend to the complex case in a straightforward manner, since now there are infinitely many choices for ${\bf P}.$ Thus, we modify the proof technique by following and extending the proof technique of \cite{Balan}.

Let ${\bf A} \in {\bf C}^{m \times n}$ be a matrix whose rows are a generic choice of vectors in ${\mathbb C}^n$. 
We note that any $k \times k$ sub-matrix of ${\bf A}$ is invertible. Let ${\cal I}, {\cal J} \subset [n]$ be two index sets of cardinality $k$ each such that ${\cal I} \cap {\cal J} = \emptyset$. Let ${\cal V} = \textrm{span}({\bf a}_k : k \in {\cal I})$ and ${\cal W} = \textrm{span}({\bf a}_k : k \in {\cal J}).$ Suppose there are two vectors ${\bf x}_{\cal I} \in {\mathbb C}^k$ and ${\bf z}_{\cal J} \in {\mathbb C}^k$ such that
$${\bf y} = |{\bf A}_{\cal I}{\bf x}_{\cal I}| = |{\bf A}_{\cal J}{\bf z}_{\cal J}|.$$

In other words,
$${\bf r} = {\bf A}_{\cal I}{\bf x}_{\cal I} = {\bf P}{\bf A}_{\cal J}{\bf z}_{\cal J}$$
for ${\bf y} = |{\bf r}|$ and for some ${\bf P} \in {\cal P}({\mathbb C}).$ We re-write this set of equations as:
\begin{equation} \label{r_in_I}
{\bf r} = {\bf A}_{\cal I}({\bf A}_{\cal I}^{[k]})^{-1}{\bf A}_{\cal I}^{[k]}{\bf x}_{\cal I} \triangleq \left[ \begin{array}{cc}
{\bf I}\\
{\bf V} \end{array} \right]\ {\bf d},
\end{equation}
and
\begin{equation} \label{r_in_J}
{\bf r} = {\bf P} {\bf A}_{\cal J}({\bf A}_{\cal J}^{[k]})^{-1}{\bf A}_{\cal J}^{[k]}{\bf z}_{\cal J} \triangleq {\bf P} \left[ \begin{array}{cc}
{\bf I}\\
{\bf W} \end{array} \right]\ {\bf e},
\end{equation}
where ${\bf I}$ is the identity matrix. We note ${\bf d} = {\bf r}_{[k]}$ and ${\bf P}_{[k]}^{[k]}{\bf e} = {\bf r}_{[k]}.$

Since the optimization procedure in (\ref{opt-ell0}) is scale-invariant and global-phase-invariant, without loss of generality we can assume $r_1 = 1$\footnote{This can be achieved by dividing both sides by $r_1$, assuming $r_1 \neq 0$. Otherwise a permutation can be applied, prior to the formulation in Equations (\ref{r_in_I}) and (\ref{r_in_J}), to the rows of ${\bf A}$ and ${\bf y}$ to make sure $r_1 \neq 0.$}. Similarly, since uniqueness is up to a global phase, we can assume $p_{11} = 1$ by absorbing the phase term into ${\bf e}$. Finally, without loss of generality we can assume $r_1, r_2, \dots, r_k$ are non-zero. We first note that at least $k$ elements in ${\bf r}$ are non-zero, since unique sparse recovery requires at least $2k$ linearly independent measurements \cite{Candes-Romberg-Tao-fft}. If only $<k$ of the measurements in ${\bf r}$ are non-zero, the $k$-sparse vector being measured lies in the null-space of a matrix that has $>k$ rows (corresponding to the 0 entries in ${\bf r}$). This would imply that these measurements are linearly dependent. Following a permutation of the rows, without loss of generality, we can assume that the first $k$ elements of ${\bf r}$ are non-zero.

We say two distinct $k$-planes $({\cal V}, {\cal W})$, both in ${\mathbb C}^m$ satisfy the distinct-phaseless-mapping property if there are non-parallel distinct vectors ${\bf r} \in {\cal V}$ and ${\bf w} \in {\cal W}$ such that $|r_k| = |w_k|$ for all $k$.

From Equation (\ref{r_in_I}), for ${\bf r} \in {\cal V}$
$$r_i = \sum_{j = 1}^k r_j v_{ij},$$
for $i>k.$ Similarly from Equation (\ref{r_in_J}), for ${\bf w} \in {\cal W}$
$$w_i = \sum_{j = 1}^k r_j / p_{jj} w_{ij}.$$

Hence if $({\cal V}, {\cal W})$ satisfies the distinct-phaseless-mapping property, there exists $p_{jj} \in {\mathbb T}({\mathbb C})$ for $j \in \{1, \dots,k\}$ and $r_2, \dots, r_k \in {\mathbb C}$ (since $r_1 = 1$) such that
\begin{equation} \label{eq-constraints}
\bigg|\sum_{j = 1}^k r_j v_{ij}\bigg| = \bigg|\sum_{j = 1}^k (r_j / p_{jj}) w_{ij}\bigg|
\end{equation}
for all $k < i \leq m.$

We consider the following variety of all tuples
\begin{equation} \label{eq-variety}
\big(({\cal V}, {\cal W}), r_2, \dots, r_k,  p_{22}, \dots, p_{kk}\big).
\end{equation}

This is locally isomorphic to ${\mathbb C}^{2k(m-k)} \times ({\mathbb C} \backslash {0})^{k-1} \times ({\mathbb T}({\mathbb C}))^{k-1},$ and Equations (\ref{r_in_I}) and (\ref{r_in_J}) characterize its dimensionality \cite{Balan}, corresponding to a real dimension of $4k(m-k) + 3k - 3.$ Next, we note that the set of 2-tuples in $Gr(k,m)^{\mathbb C} \times Gr(k,m)^{\mathbb C}$ that satisfy the distinct-phaseless-mapping property is the image of the projection onto the first factor of the variety in (\ref{eq-variety}) subject to the $m-k$ equations in (\ref{eq-constraints}) \cite{Balan}.

Similar to Lemma 3.2 in \cite{Balan}, these equations are non-degenerate for any choice of non-zero $\{r_2, \dots, r_k\},$ and $\{p_{22}, \dots, p_{kk}\}$. Furthermore, the variables $\{v_{i1}, \dots, v_{ik}\}$ and $\{w_{i1}, \dots, w_{ik}\}$ appear in exactly one equation, and thus for any fixed non-zero $\{r_2, \dots, r_k\}$ and $\{p_{22}, \dots, p_{kk}\},$ these define a subspace of ${\mathbb C}^{k(m-k)} \times {\mathbb C}^{k(m-k)}$ of real codimension $\geq m-k$. This is true for all choices, implying that the equations are independent \cite{Balan}. Therefore, the set of 2-tuples in $Gr(k,m)^{\mathbb C} \times Gr(k,m)^{\mathbb C}$ that satisfy the distinct-phaseless-mapping property have real dimension $\leq 4k(m-k) + 3k - 3 - (m -k) = 4k(m-k) + 4k-3 - m.$

If $m > 4k-3,$ then this set of 2-tuples cannot be the whole of $Gr(k,m)^{\mathbb C} \times Gr(k,m)^{\mathbb C}.$ In fact, if $m > 4k-3,$ and the measurements are generic, then the set of 2-tuples in $Gr(k,m)^{\mathbb C} \times Gr(k,m)^{\mathbb C}$ that satisfy the distinct-phaseless-mapping property has measure 0. Since there are finitely many choices for ${\cal I}$ and ${\cal J}$, a union bound argument extends the result to all possible choices of $k$ columns of ${\bf A}$. Thus, if the rows of ${\bf A} \in {\mathbb C}^{m \times n}$ are a generic choice of vectors, and $m \geq 4k-2$, no two sparse vectors with $\ell_0$ norm $\leq k$ map to the same phaseless measurements acquired using ${\bf A}$. Hence $\ell_0$ minimization can uniquely recover every $k$-sparse ${\bf x} \in {\mathbb C}^{n}$ from phaseless measurements ${\bf y} = |{\bf Ax}|.$

\section{Discussion and Conclusion} \label{sec:conc}

In this note, we considered the sparse phase retrieval problem in both the real and complex settings. For the real case, we introduced the concept of phase-generalized minimum distance of a matrix, and proved a necessary and sufficient condition for unique sparse phase retrieval up to global phase. We then showed using a combinatorial approach that a matrix, ${\bf A} \in {\mathbb R}^{m \times n},$ whose entries are chosen i.i.d. from a continuous distribution with $m \geq 2k$ can be used in conjunction with $\ell_0$ minimization to uniquely recover every $k$-sparse ${\bf x}$ from phaseless measurements ${\bf y} = |{\bf Ax}|.$ For the complex case, we showed that a sufficient condition is $m \geq 4k -2$.


For the real case, the necessary and sufficient number of measurements for sparse phase retrieval matches those when the phase information is present, as in compressed sensing. Thus, the $\pm$ phase does not seem to be as valuable for real signals. 
For the complex case, we note that the number of necessary measurements for uniqueness in the non-sparse phase retrieval problem is $(4 + o(1)) n$ \cite{Bandeira}. Thus, restricting to the full support overlap case in the sparse phase retrieval problem, this implies $(4+o(1)) k$ measurements are necessary. Hence, in the complex case, our sufficient condition matches the order of the necessary condition.
During the submission of this work, we have learned of \cite{Wang} with a seemingly different approach that produces some of the bounds given in this paper.
We also note a very recent significant contribution to the (non-sparse) phase retrieval literature \cite{Conca}, which showed that $m \geq 4n-4$ measurements are sufficient for unique phase retrieval over ${\mathbb C}$. Their techniques may potentially be extended to the sparse phase retrieval problem to improve our bounds, although we did not explore this approach following the publication of \cite{Conca}.


In this paper, we only considered the noiseless sparse phase retrieval problem. In light of necessary and sufficient information theoretic conditions from the compressed sensing literature \cite{Wang-cs, Rad}, the linear scaling with respect to $k$ is not going to hold for sparse phase retrieval in general, when the measurements are corrupted with Gaussian noise. We have not performed numerical simulations to study the scaling of the number of sufficient measurements for sparse phase retrieval in the presence of noise, due to the absence of practical algorithms for optimal recovery, which would necessitate exhaustive search over all subspaces.

Future research may focus on finding efficient algorithms for sparse phase retrieval and developing the necessary conditions for sparse phase retrieval over ${\mathbb C}.$

\end{document}